\documentclass[letterpaper, 10 pt, conference]{ieeeconf}  

\IEEEoverridecommandlockouts                              
\newtheorem{theorem}{Theorem}[section]

\newtheorem{lemma}[theorem]{Lemma}

\usepackage{floatrow}
\usepackage{fix-cm}
\usepackage{etex}

\usepackage{dblfloatfix}

\usepackage{nag}

\makeatletter
\@ifpackageloaded{xcolor}{}{%
\usepackage[table,x11names,dvipsnames,svgnames]{xcolor}%
}
\makeatother

\usepackage{colortbl}

\usepackage{graphicx}
\usepackage{wrapfig}

\usepackage{cite}

\usepackage{microtype}

\usepackage{array}
\usepackage{multirow}
\usepackage{booktabs}
\usepackage{makecell} 

\ifcsname labelindent\endcsname
\let\labelindent\relax
\fi
\usepackage[inline]{enumitem}

\usepackage{subfig}

\newenvironment{lenumerate}[2][]
{\begin{enumerate}[label=(#2\arabic*),leftmargin=0.2in,itemindent=0.15in,#1]}
{\end{enumerate}}

\setlist*[enumerate,1]{label={\itshape\arabic*)}}

\makeatletter
\newcommand{\paragraphswithstop}{%
\let\copyparagraph\paragraph%
\renewcommand\paragraph[1]{\copyparagraph{##1.}}%
}
\makeatother

\usepackage[framemethod=tikz]{mdframed}

\input{preamble/math}
\newcommand{\real}[1]{\mathbb{R}^{#1}{}}

\newcommand{\bmat}[1]{\begin{bmatrix}#1\end{bmatrix}}

\newcommand{\transpose}{^\mathrm{T}}

\DeclarePairedDelimiter{\abs}{\lvert}{\rvert}

\DeclarePairedDelimiter{\norm}{\lVert}{\rVert}

\newcommand{\vct}[1]{\mathbf{#1}}

\DeclareMathOperator{\stack}{stack}

\newcommand{\iV}[1][]{{i \in V_{#1}}}
\newcommand{\ijE}[1][]{(i,j) \in E_{#1}}

\input{preamble/utilities}

\providecommand{\cU}{\mathcal{U}}

\usepackage{units}

\newcommand{\newcolorlabel}[2]{%
  \expandafter\newcommand\csname #1\endcsname[1]{%
    \colorbox{#2}{\color{white}\textsf{\textbf{##1}}}}%
}

\newcommand{\newcommenter}[2]{%
  \expandafter\newcommand\csname #1\endcsname[1]{%
    \fcolorbox{#2}{#2}{\color{white}\textsf{\textbf{#1}}}
    {\color{#2}##1}}%
  \expandafter\newcommand\csname at#1\endcsname{%
    \fcolorbox{#2}{#2}{\color{white}\textsf{\textbf{@#1}}}
    {\color{#2}}}%
  \expandafter\newcommand\csname #1hl\endcsname[2]{%
    \colorbox{#2}{\color{white}\textsf{\textbf{#1}}}\sethlcolor{Azure2}\hl{##2}~%
    \expandafter\ifx\csname commentarrow\endcsname\relax$\leftarrow$\else \commentarrow[#2]\fi~%
    {\color{#2}##1}}%
  \expandafter\newcommand\csname #1st\endcsname[2]{%
    \colorbox{#2}{\color{white}\textsf{\textbf{#1}}}\sout{##2}~%
    \expandafter\ifx\csname commentarrow\endcsname\relax$\leftarrow$\else \commentarrow[#2]\fi~%
    {\color{#2}##1}}%
}
\newcommenter{TODO}{DodgerBlue1}
\newcommenter{rtron}{Green3}

\usepackage{comment}

\usepackage{pdfcomment}

\usepackage{soul}

\usepackage[normalem]{ulem}

\usepackage{csquotes}

\usepackage{tikz}
\usetikzlibrary{calc}
\usetikzlibrary{matrix}
\usetikzlibrary{chains}
\usetikzlibrary{shapes.geometric}
\usetikzlibrary{arrows.meta}
\usetikzlibrary{decorations.pathreplacing}

\tikzset{
  dim above/.style={to path={\pgfextra{
        \pgfinterruptpath
        \draw[>=latex,|->|] let
        \p1=($(\tikztostart)!1.5em!90:(\tikztotarget)$),
        \p2=($(\tikztotarget)!1.5em!-90:(\tikztostart)$)
        in(\p1) -- (\p2) node[pos=.5,sloped,above]{#1};
        \endpgfinterruptpath
      }
    }
  },
  dim double above/.style={to path={\pgfextra{
        \pgfinterruptpath
        \draw[>=latex,|->|] let
        \p1=($(\tikztostart)!3em!90:(\tikztotarget)$),
        \p2=($(\tikztotarget)!3em!-90:(\tikztostart)$)
        in(\p1) -- (\p2) node[pos=.5,sloped,above]{#1};
        \endpgfinterruptpath
      }
    }
  },
  dim below/.style={to path={\pgfextra{
        \pgfinterruptpath
        \draw[>=latex,|->|] let 
        \p1=($(\tikztostart)!-1em!-90:(\tikztotarget)$),
        \p2=($(\tikztotarget)!-1em!90:(\tikztostart)$)
        in (\p1) -- (\p2) node[pos=.5,sloped,below]{#1};
        \endpgfinterruptpath
      }
    }
  },
}

\tikzset{
    right angle quadrant/.code={
        \pgfmathsetmacro\quadranta{{1,1,-1,-1}[#1-1]}     
        \pgfmathsetmacro\quadrantb{{1,-1,-1,1}[#1-1]}},
    right angle quadrant=1, 
    right angle length/.code={\def\rightanglelength{#1}},   
    right angle length=2ex, 
    right angle symbol/.style n args={3}{
        insert path={
            let \p0 = ($(#1)!(#3)!(#2)$) in     
                let \p1 = ($(\p0)!\quadranta*\rightanglelength!(#3)$), 
                \p2 = ($(\p0)!\quadrantb*\rightanglelength!(#2)$) in 
                let \p3 = ($(\p1)+(\p2)-(\p0)$) in  
            (\p1) -- (\p3) -- (\p2)
        }
    }
}

\newcommand{\pgfextractangle}[3]{%
    \pgfmathanglebetweenpoints{\pgfpointanchor{#2}{center}}
                              {\pgfpointanchor{#3}{center}}
    \global\let#1\pgfmathresult  
}

\usetikzlibrary{shapes.arrows}
\newcommand{\commentarrow}[1][Azure4]{\tikz[baseline=-3pt]{\node[shape border uses incircle, fill=#1,rotate=180,single arrow, inner sep=1pt, minimum size=6pt, single arrow head extend=2pt]{};}}

\tikzset{ax/.style={-latex,line width=2pt}}

\tikzset{camera/.style={fill=Sienna1,fill opacity=0.5},%
image plane/.style={draw=RoyalBlue3,line width=2pt}}

\usepackage[english]{babel}

\newcommand{\Xopt}{X^{\textrm{opt}}}
\newcommand{\cXopt}{\mathcal{X}^{\textrm{opt}}}

\DeclareMathOperator{\Ver}{Ver}

\usepackage{tikz}

\title{\LARGE \bf
A Computational Theory of Robust Localization Verifiability in the Presence of Pure Outlier Measurements
}

\author{Mahroo Bahreinian, Roberto Tron
  \thanks{This work was supported by the National Science Foundation grant NSF NRI-1734454.}
\thanks{R. Tron is with Faculty of Mechanical Engineering and Systems Engineering,
       Boston University, Boston, MA 02215, USA
        {\tt\small tron @bu.edu}}%
\thanks{M. Bahreinian is with the Department of Systems Engineering, Boston University, Boston, MA 02215, USA 
       {\tt\small mahroobh@bu.edu}
}
}

\begin{document}

\thispagestyle{empty}
\pagestyle{empty}
\maketitle


\begin{abstract}
The problem of localizing a set of nodes from relative pairwise measurements is at the core of many applications such as Structure from Motion (SfM), sensor networks, and Simultaneous Localization And Mapping (SLAM). In practical situations, the accuracy of the relative measurements is marred by noise and outliers; hence, we have the problem of quantifying how much we should trust the solution returned by some given localization solver. In this work, we focus on the question of whether an $\ell_1$-norm robust optimization formulation can recover a solution that is identical to the ground truth, under the scenario of translation-only measurements corrupted exclusively by outliers and no noise; we call this concept \emph{verifiability}. On the theoretical side, we prove that the verifiability of a problem depends only on the topology of the graph of measurements, the edge support of the outliers, and their signs, while it is independent of ground truth locations of the nodes, and of any positive scaling of the outliers. On the computational side, we present a novel approach based on the dual simplex algorithm that can check the verifiability of a problem, completely characterize the space of equivalent solutions if they exist, and identify subgraphs that are verifiable. As an application of our theory, we provide a procedure to compute \emph{a priori} probability of recovering a solution congruent or equivalent to the ground truth given a measurement graph and the probabilities of each edge containing an outlier.

\end{abstract}


\section{Introduction}
The problem of localizing a set of agents or nodes with pairwise relative measurements can be modeled as a \textit{pose graph} \cite{posegraph}, where the nodes are associated to vertices and pairwise relative measurements are associated to edges. Typical solutions are cast as maximizing the likelihood of the relative pairwise measurements given the estimated agent poses, possibly after choosing different statistical models that lead to different cost functions to be optimized; this approach has been referred to as Pose Graph Optimization (PGO) \cite{PGO} 
Different versions of this problem have been of interest in a number of fields. In computer vision, the Structure from Motion (SfM) problem \cite{SfM} aims to recover the location and orientation of cameras, and the location of 3-D points in the scene, given an unordered collection of 2D images. In sensor networks, the nodes need to be localized from relative translation or distance measurements \cite{Biswas,Cucuringu2012,Eren2003}. In robotics, the Simultaneous Localization And Mapping (SLAM) \cite{SLAM2,Grisetti2010} problem aims to recover the pose trajectories of one or more mobile agents, while building a map of the environment, using multimodal measurements (extracted from images or inertial measurement units). In all these applications, pairwise measurements are generally corrupted by a combination of small-magnitude noise and large-magnitude outliers, due to hardware, environmental, and algorithmic factors~\cite{Yang2013}. 

The simplest and most common objective employed in PGO is the least square error \cite{l2norm,l2norm_2}, which corresponds to the assumption that measurements are affected by Gaussian noise (typically having low variance). However, the solution of least square optimization can be greatly impacted by the presence of outliers (one or two isolated outliers can bias the solution for all the nodes). In \cite{stable_camera_motion, surveyofSFM}, the authors estimate the location of the nodes (with relative direction measurements) by minimizing a least square objective function with global scale constraints through a semi-definite relaxation (SDR), while  \cite{Tron1,Tron2} solve a similar problem through constrained gradient descent; in both cases, although some theoretical analysis of the robustness of the method to noise is given, the resulting methods are not robust to outliers (due to the use of the least squares cost). To obtain robustness, a possible approach is to use a pre-processing stage (e.g., using Bayesian inference or other mechanisms) to pre-process the measurements and remove outliers, followed by PGO \cite{l_inf, estimation, zach2010, inproceedings, Yang2013}.
An alternative or complementary method is to optimize robust (ideally convex) cost functions, such as the Least Unsquared Deviation (LUD) \cite{LUD,Shapefit} or others \cite{Zach}; 
in this case, the optimization can be carried out using re-weighting techniques (such as Iterative Reweighted Least Squares, IRLS \cite{IRLS} or others \cite{6630557, 6385590}), or Alternate Direction Method of Multipliers (ADMM, \cite{ADMM,admmEx}).
In all these robust approaches, it has been shown empirically that the results are close to the ground truth even in the presence of outliers; however, there have been no published attempts to characterize, in a precise way, what kind of situations can be tolerated by the solvers. The reader should contrast this, for example, to the simple case of the median in statistics, where it is well known that such estimator is robust up to 50 percent of outliers \cite{venables2013modern, mitiche2013computational, crosillaadvanced}.


The goal of this paper is to obtain results for PGO that are similar in spirit to those available for the median in classical robust estimation theory. In order to obtain strong theoretical results on the effect of outliers alone, in this paper we focus on the case where we are interested in recovering only translations (not rotations), and there is no Gaussian noise (i.e., each measurement is either perfect, or corrupted by an outlier of arbitrarily high, but bounded, magnitude); we plan to extend our results to more realistic situations in our future work. As the objective function in the optimization, we use the least absolute value deviation ($\ell_1$-norm), which is convex and allows us to bring the extensive tools from linear optimization to our disposal.
Under these conditions, it can be empirically noticed that the robustness of the $\ell_1$ cost function leads to three possible outcomes: the solution found by the solver and the ground truth are either congruent; different, but with the same value for the cost; or drastically different. Moreover, this categorization appear to depend on where the outliers are situated, but not on their absolute magnitude.
We formalize this observation in the notion of \emph{verifiability} for a graph. Given an hypothesis for the edge support of the outliers and their sign, we can use convex optimization theory to predict whether solving the $\ell_1$ optimization problem can recover the ground truth solution, whether this can be done uniquely, and, if not, completely characterize the set of solutions, while identify which subsets of the graph can be exactly recovered.
From this, and by knowing the probability of each edge to be an outlier with a given sign, we can then compute the probability that the recovered solution is completely or partially congruent to the ground truth embedding (without knowing the actual support of the outliers). Moreover, the procedure can be extended to identify subgraphs that can be uniquely localized with high probability.

\section{Notation And Preliminaries}
In this section we formally define our measurement model, the optimization problem for localizing nodes from relative measurements, and we define the notion of \emph{verifiability}.
\subsection{Graph Model}
\begin{definition}\label{def:graph}
A sensor network is modeled as an oriented graph $G=(V,E)$, where $V=\{1,\hdots,N\}$ represents the set of sensors, and $E\subset V\times V$ represents the pairwise relative measurements; we have $\ijE$ if and only if there is a measurement between node $i \in V$ and node $j \in V$. We assume that $G$ is connected. We use $\abs{V}$, $\abs{E}$ to indicate the cardinality of the sets $V$ and $E$, respectively.

\end{definition}

\begin{definition}
An \emph{embedding} of the graph associates each node $i$ to a position $\vct{x}_i\in \real{d}$.
Mathematically, we identify an embedding with a matrix $\vct{X}_V=\bmat{\vct{x}_1 & \hdots & \vct{x}_{\abs{V}}} \in \mathbb{R}^{\abs{V}\times d}$,  with $d$ being the ambient space dimension; we denote the ground truth embedding as~$\vct{X}^\ast_V$.
\end{definition}

\begin{definition}
A measurement between node $i$ and $j$, $\ijE$, is modeled as
\begin{equation}\label{measurements}
    t_{ij}=\vct{x}^*_j-\vct{x}^*_i+\epsilon_{ij},
\end{equation}
where $\vct{x}^*_j-\vct{x}^*_i$ is the true translation between nodes $i$ and $j$, and $\epsilon_{ij}$ is a random variable for outliers with distribution
\begin{equation}\label{outliers}
    \epsilon_{ij}=\begin{cases}
                    0, & \textrm{w.p.  } 1-p_{ij}^+-p_{ij}^-\\
                    \cU^{-},  & \textrm{w.p.  } p_{ij}^{-}\\
                    \cU^+,  & \textrm{w.p.  } p_{ij}^{+}\\
                  \end{cases},
\end{equation}
where $p_{ij}^-,p_{ij}^+\in (0,1)$ are a priori probabilities of having an outlier for the edge $(i,j)$ with, respectively, negative or positive support, and $\cU^-,\cU^+$ are stochastic functions that returns a samples from a uniform distribution with arbitrary, but finite, non-zero support contained in, respectively, $\real{}_{<0},\real{}_{>0}$. If $d>1$, we assume that the entries of the vector $\epsilon_{ij}$ are i.i.d. with the same distribution \eqref{outliers}.
\end{definition}

We assume that the probabilities $p_E=\{p_{ij}\}_{\ijE}$ are known; as shown below in Theorem \ref{theorem_verifiability}, our results are valid independently of the support for $\cU^\pm$ (as long as it is finite). 

From this point on, subscripts with $V$ or $E$ refer to the vector obtained by stacking the specified quantity considered for all nodes or edges (e.g., $p_E=\stack(\{p_{ij}\}_{\ijE})$).

\begin{definition}We define the \emph{outlier support} $E_\epsilon\subset E$ such that $E_\epsilon=\{(i,j)\in E: \epsilon_{ij}\neq 0\}.$
\end{definition}

\subsection{Localization Through Robust Optimization}
Given the relative pairwise measurements $t_E$ in the graph $G$, we aim to find and characterize all the embeddings that minimize the sum of all absolute residuals, i.e.,
\begin{equation}\label{eq:l-one-opt}
\begin{aligned}
\min_{\vct{X}_V, \vct{x}_1=\vct{0}}
\sum_{(i,j)\in E} \norm{\vct{x}_j-\vct{x}_i-t_{ij}}_1\\.
\end{aligned}
\end{equation}
\subsection{Global Translation Ambiguity}
If we translate all the points in the embedding by a common translation, the cost \eqref{eq:l-one-opt} does not change, since the relative measurements also remain constant.
Without loss of generality, we fix this translation ambiguity by choosing a global reference frame such that $\vct{x}_1^\ast=\vct{x}_1=\vct{0}_d$. Since we assumed that the graph is connected (Definition~\ref{def:graph}), fixing $\vct{x}_1$ alone is sufficient to fix the global translation.
For simplicity's sake, we keep $\vct{x}_1$ as a variable in the optimization problem \eqref{eq:l-one-opt} even though it is used to fix the global translational ambiguity.

\subsection{Set of Global Optimizers $\cXopt$}\label{sec:set-global-optimizers}
We define as $\cXopt$ the set of local minimizers of \eqref{eq:l-one-opt}. Since the objective function is convex (being the sum of convex functions), we have that $\cXopt$ is convex, and is exactly given by the set of global minimizers (see \cite[Theorems 8.1, 8.3]{convexity_of_optimal_set}).  Moreover, using the fact that the value of $\vct{x}_1$ is fixed and that the graph is connected, it is possible to show that the objective function in \eqref{eq:l-one-opt} is radially unbounded, and therefore the set $\cXopt$ is compact. In fact, since \eqref{eq:l-one-opt} can be rewritten as a Linear Program (LP, see below), $\cXopt$ either reduces to a single point, or is a polyhedron with a finite number of \emph{corners} (we use this term instead of \emph{vertex} as a distinction from the individual elements of $V$).

\subsection{Verifiability}
If $E_\epsilon=\emptyset$, then $t_E$ is identical to the true measurements, and the solution of \eqref{eq:l-one-opt} would be equal to the ground truth embedding $\vct{X}^\ast_V$. However, since \eqref{eq:l-one-opt} is a robust optimization problem, the optimum value could still correspond to $\vct{X}^\ast_V$ even in the presence of outliers ($E_\epsilon\neq \emptyset$). In the latter case, however, there could be multiple minimizers all giving the same value of the $\ell_1$ objective. We start formalizing the situation with the following.
\begin{definition}\label{def:verifiability}
    A \emph{(localization) problem} is defined by a pair of a graph $G=(V,E)$ and a \emph{signed outlier support $E_\epsilon^{\pm}\subset E \times \{+,-\}$} (i.e., a subset of edges paired with signs). A problem is said to be \emph{uniquely verifiable} if  $\cXopt=\vct{X}^*_V$ (unique solution), \emph{verifiable} if $\vct{X}^*_V\in\cXopt$ (possible multiple equivalent solutions), and \emph{non-verifiable} otherwise.
  \end{definition}
  Note that, according to the definitions, uniquely verifiable problems are also verifiable.\\

In \cite[Theorem 2]{Yang2013}, the authors also introduce the concept of verifiable edge and verifiable graph; however, that work considers only the case of a single outlier ($\abs{E_\epsilon^\pm}=1$). In this work, we generalize the same notion to more general cases.


\section{Canonical LP Form And Verifiability}
In this section we perform a series of transformations to the optimization problem \eqref{eq:l-one-opt} to reduce it to a canonical, one-dimensional LP (and its dual), allowing us to deduce that particular ground-truth embeddings $X_V^\ast$ and outlier magnitudes $\epsilon_E$ do not affect the verifiability of a problem, thus ensuring that Definition~\ref{def:verifiability}, which depends only on the graph topology and the signed outlier support, is well posed.

\subsection{Canonical Form}
We first perform a change of variable so that the true embedding corresponds to the point at the origin. More in detail, we define a set of new variables $\vct{X}'_V$ such that
\begin{equation} \label{eq:canonical1}
\vct{X}_V'=X_V-\vct{X}^\ast_V,
\end{equation}
i.e., for each $\iV$ we replace $\vct{x}_i$ by $\vct{x}'_i+\vct{x}_i^\ast$. If $\vct{X}^\ast$ is an optimal point for \eqref{eq:l-one-opt}, then $\vct{X}'=\vct{0}_{\abs{V}}$ is a minimizer for the following transformed problem:

\begin{equation}\label{eq:canonical2}
\begin{aligned}
\min_{\vct{x}'_V, \vct{x}'_1=0}
\sum_{(i,j)\in E} \norm{(\vct{x}'_j+\vct{x}_j^\ast)-(\vct{x}'_i+\vct{x}_i^\ast)-(\vct{x}^*_j-\vct{x}^*_i+\epsilon_{ij})}_1 ,
\end{aligned}
\end{equation}
which reduces to
\begin{equation}\label{eq:canonical_optimization}
\begin{aligned}
\min_{\vct{x}'_V,\vct{x}'_1=0} \sum_{(i,j)\in E} \norm{\vct{x}'_j-\vct{x}'_i-\epsilon_{ij}}_1.
\end{aligned}
\end{equation}

By inspecting \eqref{eq:canonical_optimization}, we can deduce the following:
\begin{lemma}\label{lemma:canonical}
The canonical form of the optimization problem, and the definition of verifiability, do not depend on the specific value of $\vct{X}^*_V$.
\end{lemma}
\begin{proof}
Assume we have two problems with different true embeddings $\vct{X}^\ast_{V_1}$, $\vct{X}^\ast_{V_2}$, but same graph topology $G$, and the same outlier realization $\varepsilon_E$. The corresponding optimization problem in canonical form \eqref{eq:canonical_optimization} are the same, hence, also their set of solutions (after the change of variable) is the same. The rest of the claim then follows from Definition~\ref{def:verifiability}.
\end{proof}
The practical implication of Lemma~\ref{lemma:canonical} is that we can reason about the verifiability of a problem independently from the specific true positions of nodes. To simplify our discussion, for the remainder of the paper and without loss of generality we use $\vct{x}$ instead of $\vct{x}'$.

\subsection{Reduction to One-Dimensional Problems}\label{sec:dimension-reduction}


The $\ell_1$-norm $\norm{\cdot}_1 \colon \real{d}\rightarrow \real{}$ in the optimization objective can be decomposed into sums of absolute values across dimensions, i.e.,  \eqref{eq:canonical_optimization} becomes
 \begin{equation}\label{eq:summation_absolute_value}
\begin{aligned}
\min_{\vct{X}_V,[\vct{x}_1]_k=0}
& & \sum_{k=1}^{d} \sum_{(i.j)\in E} \bigl\lvert [\vct{x}_j]_k-[\vct{x}_i]_k-[\epsilon_{ij}]_k\bigr\rvert,\\
\end{aligned}
\end{equation}
where $[v]_k$ denotes the $k$-th element of a vector $v\in \real{d}$. The minimization problem \eqref{eq:summation_absolute_value} can then be decomposed into $d$ separate optimization problems, each one with a solution set $[\cXopt]_k$, $k\in \{1,\ldots, d\}$, and each one corresponding to a 1-D localization problem of the form
\begin{equation}\label{eq:directional_canonical}
    \begin{aligned}
\min_{x_V,x_1=0}
\sum_{(i.j)\in E} \mid x_j-x_i-\epsilon_{ij}\mid.
\end{aligned}
\end{equation}
We postpone to Section~\ref{sec:combine-solutions} the discussion of how to combine the results of our analysis from the different dimensions; until that section, we exclusively focus on the 1-D version of the problem.
\subsection{Canonical Linear Program Form}\label{sec:canonical}
In this section, we transform \eqref{eq:directional_canonical} into the equivalent standard Linear Program (LP) form, with a linear cost function subject to linear inequality constraints, and compute its dual. This will allow us to arrive to the conclusion that the exact magnitude of the outliers is not important in terms of verifiability, and only the signed outlier support matter.

We first introduce variables
\begin{equation}\label{Z_abs}
    Z_{ij}=\abs{x_j-x_i-\epsilon_{ij}},\;\; \forall (i,j)\in E,
  \end{equation}
to push the cost function into the constraints.
\begin{subequations}\label{eq:LP1}
    \begin{align}
 \min_{Z_E,x_V, x_1=0} 
 &  \sum_{\ijE}Z_{ij} \\
 \text{subject to }
 & x_{j}-x_{i}-\epsilon_{{ij}}\leq Z_{{ij}}, \\
 &-(x_{j}-x_{i}-\epsilon_{{ij}})\leq Z_{{ij}}, \\
 &Z_{{ij}} \geq 0,\\
& \forall \iV,\ijE.\nonumber
\end{align}
\end{subequations}
Next, in order to obtain a standard LP form, all variables must be non-negative. We therefore split each variable $x_i$ into the summation of two non-negative variables,
\begin{equation}\label{eq:LP2}
    x_i=x^+_i-x^-_i,\;  x^+_i,x^-_i\geq0.
\end{equation}
Finally, we change the inequality constraints into equality constraints by introducing the slack variables $S^+_E,S^-_E$:
\begin{subequations}\label{eq:LP3}
    \begin{alignat}{2}
\min_{Z,x,x_1=0}
 &  \sum_{(i.j)\in E}Z_{ij}, \\
 \text{subject to }
 & x^+_{j}-x^-_{j}-(x^+_{i}-x^-_{i})-\epsilon_{{ij}}+S^+_{{ij}}=Z_{{ij}},\label{eq1:LP3} \\\
& -(x^+_{j}-x^-_{j}-(x^+_{i}-x^-_{i})-\epsilon_{{ij}})+S^-_{{ij}}=Z_{{ij}}, \label{eq2:LP3}\\
& x^+_i,x^-_i,S^+_{ij},S^-_{ij},Z_{ij} \geq 0,\label{eq4:LP3}\\
& \forall \iV,\; \ijE.\nonumber
\end{alignat}
\end{subequations}

\begin{remark}[Value of $S_E$]\label{remark:value-of-ZE}
If we add constraints \eqref{eq1:LP3} and \eqref{eq2:LP3}, we obtain
\begin{equation}\label{S}
    S^+_{ij}+S^-_{ij}=2Z_{ij}.
\end{equation}
Moreover, from \eqref{Z_abs} and \eqref{S},
\begin{equation} \label{Z}
    (S^+_{ij},S^-_{ij})=\begin{cases}
      (2Z_{ij},0),& \textrm{if } Z_{ij}=-(x_j-x_i-\epsilon_{ij})\\
      (0,2Z_{ij}),& \textrm{if } Z_{ij}=x_j-x_i-\epsilon_{ij}.
    \end{cases}
\end{equation}
\end{remark}
We can also form the dual optimization problem of \eqref{eq:LP3},
\begin{subequations}\label{eq:dual}
\begin{alignat}{2}
& \max_{P_{ij}^+, P_{ij}^-}
& & \sum_{(i,j) \in E}{\epsilon_{ij}(P_{ij}^+-P_{ij}^-)}, \label{eq1:dual}\\
& \text{subject to}
& & \sum_{j,(j,i)\in E} (P_{ji}^+-P_{ji}^-)-\sum_{j,(i,j)\in E} (P_{ij}^+-P_{ij}^-)= 0,\;\label{eq2:dual}\\
& & & -P_{ij}^+-P_{ij}^-\leq 1,\label{eq3:dual}\\
&&& P_{ij}^+,P_{ij}^-\leq 0,\label{eq4:dual}\\
&&& \forall\iV,\ijE,\nonumber
\end{alignat}
\end{subequations}
where $P^+_{ij}$ is the dual variable associated to constraint \eqref{eq1:LP3}, and $P^-_{ij}$ is the dual variable associated to constraint \eqref{eq2:LP3}.

\begin{remark}[Strong duality and verifiability]\label{remark:strong-duality-Z-P}
Assume that the localization problem $(G, E_\epsilon)$ is verifiable or uniquely verifiable. Then, the origin is primal optimal, i.e., $\vct{0}_{\abs{V}}\in \cXopt$, and from \eqref{Z_abs}, we have that, at the primal optimal solution $(X^*=0,Z_E^*,S_E^{+*},S_E^{-*})$:
\begin{equation}\label{eq:Z_verifiable}
    \sum_{(i,j)\in E}{Z^*_{ij}}=\sum_{(i,j)\in E}{\abs{\epsilon_{ij}}}=\sum_{(i,j)\in E_\epsilon^\pm}{\abs{\epsilon_{ij}}};
  \end{equation}
  note that, in the last equality, the sum is only over edges in the outlier support.
  
  If a linear programming problem has an optimal solution, so does its dual, and the respective optimal costs are equal; this is known as the strong duality property \cite[Theorem 4.4]{LP}. Combining this observation with \eqref{eq:Z_verifiable}, we have that, for a dual optimal solution $(P_E^{+*},P_E^{-*})$,
\begin{equation}\label{eq:strong-duality-Z-P}
    \sum_{\ijE}Z^*_{ij}=\sum_{\ijE[\epsilon]^\pm}{\epsilon_{ij}(P^{+*}_{ij}-P^{-*}_{ij})}=\sum_{(i,j) \in E_\epsilon^\pm}\abs{\epsilon_{ij}}.
\end{equation}
\end{remark}
\begin{remark}[Discrete optimal solution for dual variables]\label{remark:discrete-P}
  Note that constraints \eqref{eq3:dual} and \eqref{eq4:dual}, together with \eqref{eq:strong-duality-Z-P} imply that the dual optimal solution is given by $(P_{ij}^{+*},P_{ij}^{-*})\in \{(-1,0),(0,-1)\}$, for all $\ijE[\epsilon]^\pm$ (i.e., there are two discrete cases for each edge with outliers, and the selection depends on the sign of $\epsilon_{ij}$), and $-1\leq P_{ij}^{+*},P_{ij}^{-*}\ \leq 0$ for the remaining edges.
\end{remark}

These remarks allow us to prove the following.
 \begin{lemma}
 \label{lemma:scale}
 For a fixed outlier support $E_\epsilon$, if we change the scale of the outliers by positive factor, the verifiability of the graph does not change.
 \end{lemma}
 \begin{proof}
   Assume that the localization problem $(G,E_\epsilon^\pm)$ is verifiable or uniquely verifiable, and that $(X_V^*=0,Z_E^*,S_E^*)$ is a primal optimal solution, while $(P_E^{*+},P_E^{*-})$ is a dual optimal solution.
   If we replace each outlier $\epsilon_{ij}$ with a positively scaled version $u_{ij}\epsilon_{ij}$, $u_{ij}>0$, $\ijE$ (the case $u_{ij}=0$ is excluded, otherwise the outlier support would change), the cost function in \eqref{eq:dual} changes, but not the constraints, so $(P_E^{*+},P_E^{*-})$ is still a dual feasible solution. Considering the second equality in \eqref{eq:strong-duality-Z-P} from Remark~\ref{remark:strong-duality-Z-P} together with Remark~\ref{remark:discrete-P}, we have that the new dual cost after rescaling is
   \begin{equation}\label{eq:rescaled-cost-P}
     \sum_{\ijE}{u_{ij}\epsilon_{ij}(P^{+*}_{ij}-P^{-*}_{ij})}=\sum_{{\ijE}_\epsilon^\pm}u_{ij}\abs{\epsilon_{ij}}.
   \end{equation}

   At the same time, the solution $(X_V^*=0,\{u_{ij}Z_{ij}^*\}_{\ijE},\{u_{ij}S_{ij}^*\}_{\ijE})$ is primal feasible, and the corresponding cost is
   
   \begin{equation}\label{eq:rescaled-cost-Z}
    \sum_{\ijE}u_{ij}Z^*_{ij}=\sum_{(i,j) \in E}u_{ij}\abs{\epsilon_{ij}}.
  \end{equation}

  From \eqref{eq:rescaled-cost-P} and \eqref{eq:rescaled-cost-Z} together with strong duality, we can therefore conclude that $(X_V^*=0,\{u_{ij}Z_{ij}^*\}_{\ijE},\{u_{ij}S_{ij}^*\}_{\ijE})$ (respectively, $(P_E^{*+},P_E^{*-})$) is primal (respectively, dual) optimal.
  This shows that $X_V^*=0$ is an optimal solution, and the rescaled problem is again verifiable; hence, one problem is verifiable if and only if all the positive scaled versions are also verifiable.
\end{proof}
Combining lemmata \ref{lemma:canonical} and \ref{lemma:scale} we have the following:
\begin{theorem}
\label{theorem_verifiability}
The notion of verifiability depends only on the graph topology $G$, the support of the outliers $E_\varepsilon$ , and the sign of the outliers.
\end{theorem}

Technically speaking, the proof above does not cover the case of unique verifiability, in the sense that the they do not exclude the case where a verifiable problem might become uniquely verifiable after rescaling (or viceversa). We are investigating this issue in our current work.


\section{Verifiability computation}
\subsection{Linear Programming}
In this section, we discuss how the dual simplex algorithm can be used to compute the verifiability of a given problem. As a result of the previous section, for our analysis, the values of $\epsilon_E$ can be choosen randomly, as long as they have the correct edge support $E_\varepsilon^\pm$. We start by rewriting the LP \eqref{eq:LP3} in matrix form:

\begin{equation}\label{eq:Opt Matrix Form}
        \begin{aligned}
& \min_q
& &  c\transpose q \\
& \text{subject to}
& & Aq=b \\
&&& q \geq 0.
\end{aligned}
\end{equation}
The vector $c=\stack(\vct{0}_{2\abs{V}},\vct{1}_{\abs{E}},\vct{0}_{2\abs{E}})$ contains the set coefficients in the cost function, while $A\in \{0,1,-1\}^{2\abs{E}\times (2\abs{V}+3\abs{E})}$, and  $b=\left[\begin{smallmatrix} 1\\-1\end{smallmatrix}\right] \otimes \epsilon_E$ defines the constraints (where $\otimes$ denotes the Kronecker's product). Finally, the vector $q=\stack(x_V^+,x_V^-,Z_E,S_E^+,S_E^-) \in \mathbb{R}^{2\abs{V}+3\abs{E}}$ contains the decision variables.

Given the standard form of the optimization problem \eqref{eq:Opt Matrix Form}, we can use the dual simplex algorithm \cite{LP} to find all the corners of the set of minimizers $\cXopt$. The algorithm and its application to our problem are summarized next.



 \subsection{Localization Via the Dual Simplex Method}
The dual simplex method is based on the following concepts:
\begin{enumerate}
    \item Basic variables (BVs): a subset of variables ($q_B$), that, together with the constraints, defines the current candidate solution in the algorithm. Non-basic variables (NBV) are always zero.
    \item Simplex tableau: a $(2\abs{E}+1)\times (2|V|+3|E|-1)$ array where
    \begin{itemize}
        \item The zeroth column represents the value of the set of basic variables ($q_B$). It is initialized with the vector $b$.
        \item The zeroth row contains the \emph{reduced costs}, which are defined as the penalty cost for introducing one unit of the variable $q_i$ to the cost. These are initialized with the vector $c$.
        \item Columns one to $2(|V|-1)+3|E|$ are each one associated with one variable, where we excluded the columns corresponding to $x^+_1,x^-_1$, since $x_1$ is fixed in the optimization. These columns are initialized with the matrix $A$.
    \end{itemize}
\end{enumerate}

For our initial estimated solution, we set all variables to zero except the slack variables; as a result, our initial BVs correspond to the set of slack variables, while the rest are NBVs. See Fig.~\ref{tab:initial_tableau} for an illustration of the initial tableau.
\begin{figure}
\begin{tabular}{c|ccccll|}
\multicolumn{1}{c}{\rule[-0.7em]{0pt}{0pt}} & $0$-th col. & \multicolumn{1}{l}{$x^+_V$} & \multicolumn{1}{l}{$x^-_V$} & \multicolumn{1}{l}{${Z_E}$}    & $S_E^+$     & \multicolumn{1}{c}{$S_E^-$}     \\ \cline{2-7}
\rule[1.3em]{0pt}{0pt} $0$-th rowx               & 0                         & $\vct{0}_V$            & $\vct{0}_V$            & $\vct{1}_E$                            &$\vct{0}_E$                   & $\vct{0}_E$                   \\ 
$q_{B(1)}$           & $b(1)$ &       $\mid$              &       $\mid$                         &          $\mid$                                    &        $\mid$                          &         $\mid$                          \\
$q_{B(2)}$          & $b(2)$ & $a_{x^+_V}$          & $a_{x^-_V}$              & $a_{Z_E}$                                 &$a_{S^+_E}$                 & $a_{S^-_E}$                      \\
$\vdots$        & $\vdots$  &        $\mid$                    &       $\mid$                      &      $\mid$                                           &         $\mid$                         &        $\mid$                            \\
  $q_{B(2\abs{E})}$  & \rule[-0.7em]{0pt}{0pt}$b(2\abs{E})$   &         &                          &                                             &                             &                          \\
  \cline{2-7}  
\end{tabular}
\caption{Initial simplex tableau, with labeled rows and columns}
\label{tab:initial_tableau}
\end{figure}


A typical iteration starts with some basic variables containing negative elements, and all reduced costs non-negative. For instance, in Fig.~\ref{tab:initial_tableau}, the initial BVs are selected to be slack variables where $S^+_{ij}=\epsilon_{ij}$ and $S^-_{ij}=-\epsilon_{ij}$, hence, there are some negative initial BVs, while all reduced costs are non-negative (as all elements of vector $c$ are non-negative). These two properties are always maintained by the algorithm from one iteration to the next.

The iterations of the algorithm then follow these steps:
\begin{enumerate}
    \item\emph{Check for termination due to optimality:} Examine the elements of zeroth column (which constitutes the basic set). If all of them are non-negative, we have an optimal basic solution and the algorithm terminates.
    \item\emph{Choose pivot row:} Find some $\nu$ such that $[q_{B}]_\nu<0$.
    \item\emph{Check for termination due to unbounded solution:} Considering the $\nu$-th row of the tableau, with elements $r_1,\hdots,r_{2(\abs{V}-1)+3\abs{E}}$, if all the elements of the row are non-negative, the optimal dual cost is $+\infty$ and algorithm terminates. Since the set of minimizers $\cXopt$ in our problem is bounded (see Section~\ref{sec:set-global-optimizers}), this condition is never encountered in our application.
    \item\emph{Choose pivot column:} For each $i$ such that $r_i<0$, compute the ratio $\Bar{c_i}/ |r_i|$ where $\Bar{c_i}$ is the reduced cost of variable $q_i$ and let $j$ be the index of a column that correspond to the smallest ratio. 
    \item\emph{Pivoting:} Remove the variable $[q_B]_\nu$ from the basis, and have variable $q_j$ take its place. Add to each row of the tableau a multiple of the $\nu$-th row (pivot row) so that $r_j$ (the pivot element) becomes $1$ and all other entries of the pivot column become $0$. As a result, the total cost is reduced by the reduced cost $\bar{c}_j$.
    \item  Repeat the algorithm from step 2 until all elements of $q_B$ are non-negative or the algorithm otherwise terminates.
\end{enumerate}

 After solving the simplex tableau, we get the basic optimal solution, which contains non-negative elements, together with non-negative reduced costs. The solution of the dual simplex algorithm is an optimal solution for \eqref{eq:Opt Matrix Form}, and is a corner point of the feasible region (Theorem 2.3, \cite{LP}). If we have multiple optimal solutions (i.e., $\cXopt$ is not a singleton), there will be multiple other corners with the same cost. 

 Hence, it is of interest to computationally enumerate all the corners of $\cXopt$, as discussed next.
 \subsection{Characterizing $\cXopt$ And Verifiability}
 The LP problem \ref{eq:Opt Matrix Form} can have multiple optimal solutions only when two conditions are met \cite{uniqueness}:
\begin{enumerate}
\item There exists a non-basic variable with zero reduced cost. Pivoting this variable into the basis would not change the value for the cost function.
\item There exists a degenerate basic solution, i.e. some basic variables are equal to zero.
\end{enumerate}

If the two conditions above are met, the corners in $\cXopt$ can be enumerated
using a depth first search \cite{dfs}:
\begin{enumerate}
    \item Prepare a queue $Q$ of corners to visit, with the corresponding tableau, and initialize it with the current solution found by the dual simplex algorithm,
    \item For each corner in $Q$ and its associated tableau,
    \begin{enumerate}
    \item Choose $C_{col}$ as the set of columns associated to non-basic variables with zero reduced cost, for all  $j\in C_{col}$,
    \begin{enumerate}
    \item Choose $C_{row}$ as the set of elements of the $j$-th pivot column which are positive,
    \item For $i\in C_{row}$, we perform the pivoting, so that the pivot element in $i$-th row and $j$-th column becomes $1$ and all other entries of the pivot column become~$0$, 
    \item Add the current corner to the queue $Q$, if is not in it already,
       \end{enumerate} 
    \end{enumerate}
    \item Go to step 2 until the queue $Q$ is empty.
\end{enumerate}

\begin{remark}\label{shifitng cost}
In terms of our localization problems, the pivoting variables and the motion from one corner of $\cXopt$ to another can be given a physical interpretation.
We defined as $Z_{ij}$ the cost of edge $(i,j)$. Assuming we have a verifiable graph, from \eqref{eq:Z_verifiable}, the cost of edge $(i,j)$ is equal to $\abs{\epsilon_{ij}}$. When we move (pivot) to another corner with the same cost, the set of basic variables changes, but the value of all the other variables remains the same. So, if a non-basic variable takes the place of basic variables from the set $x^+_V$ or $x^-_V$, it does not produce a new optimal embedding (because such variables where already equal to zero). If a pivoting variable takes the place of non-zero basic variable $Z_{ij}$, then $Z_{ij}$ becomes zero, which means the cost of edge $(i,j)$ changes to zero, and if $\epsilon_{ij}\neq0$ then from \eqref{eq:Z_verifiable}, $x_i$ and $x_j$ are not equal to zero anymore. As the value of cost function remains the same, the loss of cost of edge $(i,j)$ must be compensated with the costs of the rest of the edges. If we pivot a non-basic variable to the non-zero basic variable $S_{ij}^+$ or $S_{ij}^-$, from \eqref{S}, it implies the value of $Z_{ij}$ becomes zero which means the cost of edge $(i,j)$ changes to zero. So, pivoting non-basic variable in order to find alternative solutions means shifting the cost of outliers from one edge to the others.
\end{remark}

There are three cases for the set of optimal solutions, $\cXopt$:
\begin{enumerate}
\item\emph{Uniquely verifiable solution:}
Pivoting new variables to the basis does not result in new corner point; we therefore have a unique optimal solution $\cXopt=\{\vct{0}_{V}\}$, and from \eqref{eq:canonical1} we conclude that the resulting embedding is congruent to the ground truth.
\item\emph{Verifiable (non-unique) solution:} We have multiple optimal solutions, including the origin ($\vct{0}_{V}\in\cXopt$); hence, there are multiple optimal embeddings, with one of them being congruent to the ground truth.
\item\emph{Non-verifiable:} In this case, $\vct{0}_{V}\notin\cXopt$, and the ground truth embedding is not an optimal solution.  
\end{enumerate}
\subsection{Combining Solutions From Multiple Dimensions}
\label{sec:combine-solutions}

In Section~\ref{sec:dimension-reduction}, we reduced one $d$-dimensional optimization problem of the form \eqref{eq:canonical_optimization} to $d$ 1-D optimization problems of the form \eqref{eq:directional_canonical}. Now, we need to combine the optimal solutions of all dimensions to characterize the $d$-dimensional optimal solution. Let $[\cXopt]_k$ represents the set of optimal solutions for the LP \eqref{eq:LP1} of dimension $k$. The value of the cost function \eqref{eq:LP1} is the same for all corner points in $[\cXopt]_k$. Due to this fact, we can pick a 1-D corner point from each set $[\cXopt]_k$, $k\in \{1,\hdots,d\}$, and combine them to build a $d$-dimensional corner point:
\begin{equation}
    x^{opt}=\stack(X^{opt}_1,\hdots, X^{opt}_d), \; X^{opt}_k \in [\cXopt]_k.
\end{equation}
Let $\abs{[\cXopt]_k}$ represents the cardinality of the set $[\cXopt]_k$; then, we have $N=\prod_{k=1}^d \bigl\lvert [\cXopt]_k\bigr\rvert$ $d$-dimensional corner points.  To  have a unique verifiable graph, we therefore need all the individual 1-D problems to be also unique verifiable, i.e. $\abs{[\cXopt]_k}=1$ for all $k\in \{1,\hdots,d\}$.

\subsection{Maximal verifiable components}\label{sec:maximal-components}
If for all corners a subset of components $V'$ in the solution are always zero (i.e., $[X^{opt}_k]_V'=0$ for all $k$), then the position of those particular nodes, and all their relative positions, are congruent to the true embedding. As a consequence, also all their relative costs are the same. Hence, while the entire problem $G,E^\pm_\epsilon$ is not verifiable, the sub-problem $G',E_\epsilon'^\pm$, where $G'=(V',E')$, $E'=\{(i,j)\in E: i,j\in V'\}$ is verifiable. We call the maximal connected components of $G'$ defined in this way the \emph{maximal verifiable components} of $G$.

\section{Verifiability Probability}
Given a tuple $(G,E_\epsilon^\pm)$ of a graph and a signed outlier support, we can define a function that indicate if the associated localization problem is verifiable,
\begin{align}
    \Ver (G,E_\epsilon^\pm)&=\begin{cases} 1&\textrm{ if }0\in\Xopt\\
    0 &\textrm{otherwise}
    \end{cases}
\end{align}
This function can be implemented by using the dual simplex algorithm discussed above.

Given the edge outlier probabilities $p^\pm_E$ defined in \eqref{outliers}, we can take the expectation of $\Ver(G,\cdot)$ over different outlier realizations, and hence characterize the a priori probability of recovering a localization that is cost-equivalent to the true one, without knowing the exact value or support of the outliers.
\begin{definition}
  We define the \emph{verifiability probability} $p_{\Ver}$ as the probability of recovering a solution whose cost is the same as the ground truth, i.e., $p_{\Ver}=\mathbb{E}_{\epsilon}[\Ver(G,E_\epsilon^\pm)]$, where $\mathbb{E}_{\epsilon}[\cdot]$ is the expectation over all the realizations of outliers. 
\end{definition}

The interpretation of this number is the a priori probability that the ground truth embedding $X_V^*$ belongs to $\cXopt$, the set of minimizers of \eqref{eq:l-one-opt}.
For instance, if we assume the edge positive outlier probability is $p^+$, and the edge negative outlier probability is $p^-$, then we can define 
$p(\epsilon_E)=(p^+)^{\abs{E^+_\epsilon}}(p^-)^{\abs{E^-_\epsilon}}(1-p^+-p^-)^{(\mid E \mid-\mid E^-_\epsilon\mid-\mid E^+_\epsilon\mid)}$
and $p_{\Ver}=\mathbb{E}_{\epsilon}[\Ver(E,E_\epsilon^\pm)]$.

Note that an analogous quantity could be computed for unique verifiability, although we would need to expand our results to make this rigorous (see comments immediately after Theorem~\ref{theorem_verifiability}).
Moreover, a similar concept could be extended to each individual edge, or any arbitrary subset of edges, by asking whether they are part of a maximal verifiable component (Section~\ref{sec:maximal-components}). Nonetheless, a formal exploration of these concepts is out of the scope of the present paper.
\section{Numerical Examples}
In this section we apply our theory and algorithm\footnote{The algorithm is implemented in MATLAB at thttps://github.com/Mahrooo/Robust-Localization-Verifiability.git} to a simple graph with 5 nodes and 10 edges, $\vct{X}_V \in \mathbb{R}^{5 \times 2} $ (Fig. ~\ref{fig:exp1}).
We start with the case where three relative measurements in first coordinate are outliers 
 and all other measurements (Fig. ~\ref{fig:exp1true})
are accurate. In this example, positive and negative outlier have the same probability $p^+=p^-=\frac{1}{2}p$.
After solving the optimization problem associated to this graph, we find three different embeddings that represent the corners of $\cXopt$; these are shown in Fig. ~\ref{fig:exp1c1}, \ref{fig:exp1c2} and \ref{fig:exp1c3}.
\begin{figure}[t]
\centering
\subfloat[Ground truth embedding]{\label{fig:exp1true}\includegraphics[scale=0.6]{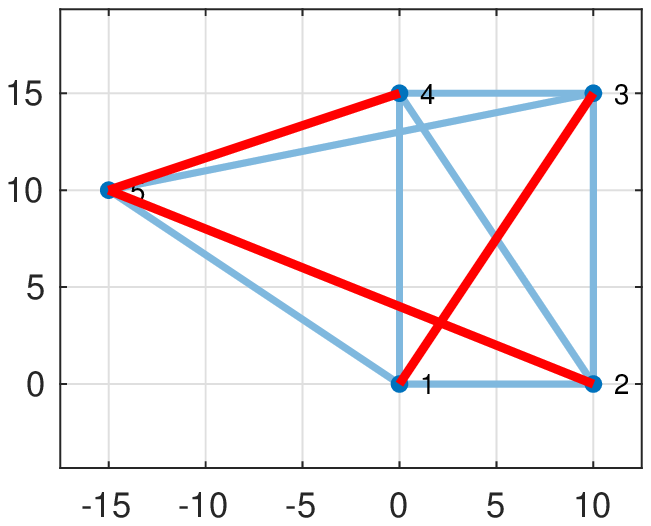}} 
\hfill
\subfloat[Optimal embedding $\cXopt_1$]{\label{fig:exp1c1}\includegraphics[scale=0.6]{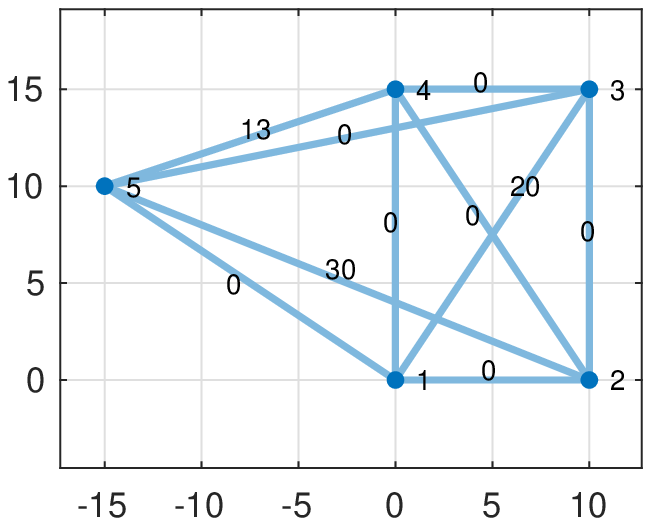}}\\
\end{figure}
\\
\begin{figure}[t]
\subfloat[Optimal embedding $\cXopt_2$]{\label{fig:exp1c2}\includegraphics[scale=0.6]{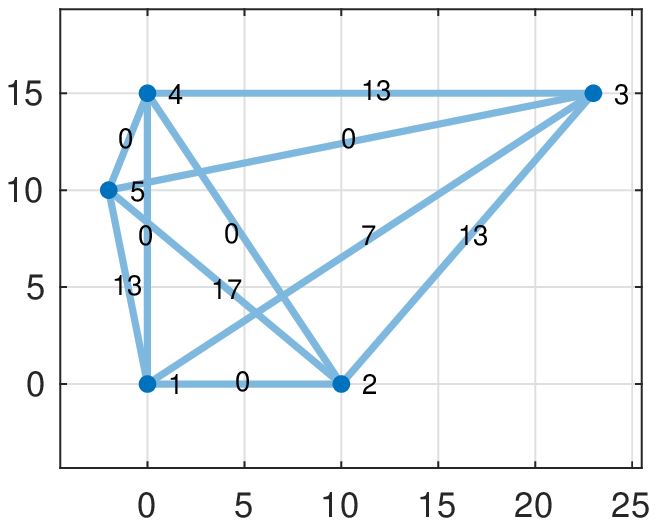}}
\hfill
\subfloat[Optimal embedding $\cXopt_3$]{\label{fig:exp1c3}\includegraphics[scale=0.6]{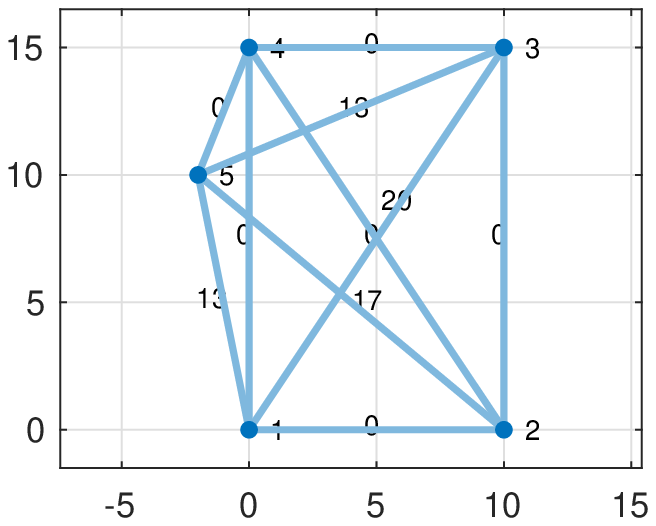}}
\caption{Verifiable graph with 5 nodes and 10 edges, 3 edges are outliers are shown by red color in Fig. ~\ref{fig:exp1true}, the cost of each edge is shown on each and the cost of optimal solution for all embeddings are equal to the cost of the ground truth embedding which is 63}
\label{fig:exp1}
\end{figure}
In Fig. ~\ref{fig:exp1c1}, the resulted embedding is identical to the ground truth embedding, which means that $\vct{x}_V\in\cXopt$, and the graph is verifiable. However, since we have multiple solution, the graph is not uniquely verifiable. In the figures, the cost of associated to each edge is shown; it can be seen that different corners shift the cost to different edges, although their sum remains the same. The locations of nodes $V'=\{1,2,4\}$ are identical to their ground truth locations,
and the costs of edges $E'=\{(4,1),(2,1),(2,4)\}$ remain the same in all embeddings, so the subgraph $G=(V',E')$ is a maximal verifiable component.

Assuming that the edge outlier probability $p_{ij}$ is $\frac{1}{2} p$ for all edges $\ijE$, then for our graph in this example the verifiability probability for this graph can be evaluated as 
\begin{equation}
 \begin{aligned}
& p_{\Ver}=(1-p)^{10}+ 20(\frac{p}{2})(1-p)^9+180(\frac{p}{2})^2(1-p)^8\\
& +920(\frac{p}{2})^3(1-p)^7+2680(\frac{p}{2})^4(1-p)^6+4524(\frac{p}{2})^5(1-p)^5\\
&  +4560(\frac{p}{2})^6(1-p)^4+2820(\frac{p}{2})^7(1-p)^3\\
&  +1080(\frac{p}{2})^8(1-p)^2+240(\frac{p}{2})^9(1-p)+24(\frac{p}{2})^{10},
\end{aligned}
\end{equation}
where the coefficients come from Table \ref{tab:edge-support-cases}.
\begin{figure}[b]
    \centering
    \includegraphics[scale=0.6]{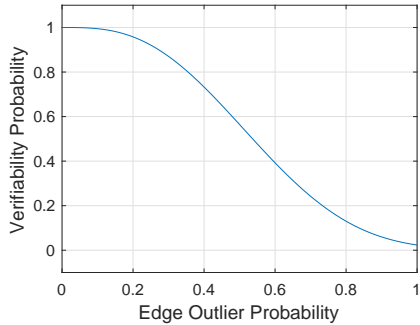}
    \caption{Verifiability probability for the graph in Fig ~\ref{fig:exp1}}
    \label{fig:probability}
\end{figure}
As shown in Fig. \ref{fig:probability}, if $p=0$ we have a verifiable graph with probability $p_{\Ver}=1$; as we increase the probability of more edges to be outliers, the probability of having access to the verifiable graph decreases.


\begin{table}[t]
   \centering
   \begin{tabular}{p{0.25\linewidth}p{0.25\linewidth}p{0.25\linewidth}}
   \toprule
    \#outliers, $\abs{E_\epsilon^\pm}$ & \#possible combinations, $\abs{E} \choose \abs{E_\epsilon^\pm}$ &  \#verifiable combinations\\
    \midrule
    0 & 1 &1\\
    1 & 20 &20\\
    2 & 180 &180\\
    3 & 960 &920\\
    4 & 3360 &2680\\
    5 & 8064 &4524 \\
    6 & 13440 &4560\\
    7 & 15360 &2820 \\
    8 & 11520  & 1080\\
    9 & 5120  & 240\\
    10 & 1024  & 24\\
    \bottomrule
   \end{tabular}
   \caption{Verifiability analysis for all possible cases of outlier supports $E_\varepsilon^\pm$}
   \label{tab:edge-support-cases}
\end{table}
\section{Conclusions And Future Works}
In this work, we consider the estimation of an embedding for nodes with relative translation measurements affected by outliers (but no noise) through the minimization of an $\ell_1$-norm cost function. We introduce the notion of verifiability, which characterizes when we can expect to recover a solution with cost equal to the true one; we show that the concept of verifiability depends only on the topology of the network and where the outliers are placed, and we also provide a way to compute it using the dual simplex method. From a more practical standpoint, we define the verifiability probability, which characterizes the a priori reliability that can be expected from a given measurement graph (given a priori probabilities of outliers for each edge). There are many possible directions for our future work. First, we plan to include the effects of amplitude-limited noise to our measurement models, and study its effect of noise on our results; concurrently, we will study different cost functions, such as the Huber-loss function and piece-wise linear loss functions.



\addtolength{\textheight}{-10cm}

\bibliographystyle{ieee}
\bibliography{biblio/IEEEConfFull,biblio/IEEEFull,biblio/OtherFull,references}

\begin{thebibliography}{10}\itemsep=-1pt

\bibitem{6630557}
P.~{Agarwal}, G.~D. {Tipaldi}, L.~{Spinello}, C.~{Stachniss}, and W.~{Burgard}.
\newblock Robust map optimization using dynamic covariance scaling.
\newblock In {\em {IEEE} International Conference on Robotics and Automation},
  pages 62--69, May 2013.

\bibitem{uniqueness}
G.~Appa.
\newblock On the uniqueness of solutions to linear programs.
\newblock {\em Journal of the Operational Research Society}, 53(10):1127--1132,
  2002.

\bibitem{l2norm}
M.~Arie-Nachimson, S.~Z. Kovalsky, I.~Kemelmacher-Shlizerman, A.~Singer, and
  R.~Basri.
\newblock Global motion estimation from point matches.
\newblock In {\em Second International Conference on 3D Imaging, Modeling,
  Processing, Visualization \& Transmission}, pages 81--88. IEEE, 2012.

\bibitem{convexity_of_optimal_set}
A.~Beck.
\newblock {\em Introduction to nonlinear optimization: theory, algorithms, and
  applications with MATLAB}, volume~19.
\newblock Siam, 2014.

\bibitem{LP}
D.~Bertsimas and J.~N. Tsitsiklis.
\newblock {\em Introduction to linear optimization}, volume~6.
\newblock Athena Scientific Belmont, MA, 1997.

\bibitem{Biswas}
P.~Biswas, T.-C. Lian, T.-C. Wang, and Y.~Ye.
\newblock Semidefinite programming based algorithms for sensor network
  localization.
\newblock {\em ACM Transactions on Sensor Networks (TOSN)}, 2(2):188--220,
  2006.

\bibitem{ADMM}
S.~Boyd, N.~Parikh, E.~Chu, B.~Peleato, J.~Eckstein, et~al.
\newblock Distributed optimization and statistical learning via the alternating
  direction method of multipliers.
\newblock {\em Foundations and Trends{\textregistered} in Machine learning},
  3(1):1--122, 2011.

\bibitem{crosillaadvanced}
F.~Crosilla, A.~Beinat, A.~Fusiello, E.~Maset, and D.~Visintini.
\newblock Advanced procrustes analysis models in photogrammetric computer
  vision.

\bibitem{Cucuringu2012}
M.~Cucuringu, Y.~Lipman, and A.~Singer.
\newblock Sensor network localization by eigenvector synchronization over the
  euclidean group.
\newblock {\em ACM Transactions on Sensor Networks (TOSN)}, 8(3):19, 2012.

\bibitem{PGO}
F.~Endres, J.~Hess, N.~Engelhard, J.~Sturm, D.~Cremers, and W.~Burgard.
\newblock An evaluation of the rgb-d {SLAM}system.
\newblock In {\em {IEEE} International Conference on Robotics and Automation},
  volume~3, pages 1691--1696, 2012.

\bibitem{Eren2003}
T.~Eren, W.~Whiteley, A.~S. Morse, P.~N. Belhumeur, and B.~D. Anderson.
\newblock Sensor and network topologies of formations with direction, bearing,
  and angle information between agents.
\newblock In {\em {IEEE} International Conference on Decision and Control},
  volume~3, pages 3064--3069. IEEE, 2003.

\bibitem{admmEx}
T.~Erseghe.
\newblock A distributed and scalable processing method based upon admm.
\newblock {\em {IEEE} Signal Processing Letters}, 19(9):563--566, 2012.

\bibitem{l2norm_2}
V.~M. {Govindu}.
\newblock Combining two-view constraints for motion estimation.
\newblock In {\em {IEEE} Conference on Computer Vision and Pattern
  Recognition}, volume~2, pages II--II, Dec 2001.

\bibitem{Grisetti2010}
G.~{Grisetti}, R.~{Kummerle}, C.~{Stachniss}, and W.~{Burgard}.
\newblock A tutorial on graph-based slam.
\newblock {\em {IEEE} Intelligent Transportation Systems Magazine},
  2(4):31--43, winter 2010.

\bibitem{Shapefit}
P.~Hand, C.~Lee, and V.~Voroninski.
\newblock Shapefit: Exact location recovery from corrupted pairwise directions.
\newblock {\em Communications on Pure and Applied Mathematics}, 71, 06 2015.

\bibitem{SfM}
R.~Hartley and A.~Zisserman.
\newblock {\em Multiple View Geometry in Computer Vision}.
\newblock Cambridge University Press, New York, NY, USA, 2 edition, 2003.

\bibitem{IRLS}
P.~Holland and R.~E.~Welsch.
\newblock Robust regression using iteratively reweighted least-squares.
\newblock {\em Communications in Statistics-theory and Methods - COMMUN
  STATIST-THEOR METHOD}, 6:813--827, 01 1977.

\bibitem{posegraph}
R.~{Kümmerle}, G.~{Grisetti}, H.~{Strasdat}, K.~{Konolige}, and W.~{Burgard}.
\newblock G2o: A general framework for graph optimization.
\newblock In {\em {IEEE} International Conference on Robotics and Automation},
  pages 3607--3613, May 2011.

\bibitem{estimation}
G.~Lerman and Y.~Shi.
\newblock Estimation of camera locations in highly corrupted scenarios: All
  about that base, no shape trouble.
\newblock In {\em {IEEE} Conference on Computer Vision and Pattern
  Recognition}, pages 2868--2876, 06 2018.

\bibitem{mitiche2013computational}
A.~Mitiche.
\newblock {\em Computational analysis of visual motion}.
\newblock Springer Science \& Business Media, 2013.

\bibitem{l_inf}
P.~Moulon, P.~Monasse, and R.~Marlet.
\newblock Global fusion of relative motions for robust, accurate and scalable
  structure from motion.
\newblock In {\em {IEEE} International Conference on Computer Vision}, pages
  3248--3255, 2013.

\bibitem{stable_camera_motion}
O.~Ozyesil, A.~Singer, and R.~Basri.
\newblock Stable camera motion estimation using convex programming.
\newblock {\em SIAM Journal on Imaging Sciences}, 8, 07 2014.

\bibitem{surveyofSFM}
O.~Ozyesil, V.~Voroninski, R.~Basri, and A.~Singer.
\newblock A survey on structure from motion.
\newblock {\em Acta Numerica}, 26, 01 2017.

\bibitem{SLAM2}
N.~{Sünderhauf} and P.~{Protzel}.
\newblock Switchable constraints for robust pose graph slam.
\newblock In {\em {IEEE} International Conference on Intelligent Robots and
  Systems}, pages 1879--1884, Oct 2012.

\bibitem{6385590}
N.~{Sünderhauf} and P.~{Protzel}.
\newblock Switchable constraints for robust pose graph slam.
\newblock In {\em {IEEE} International Conference on Intelligent Robots and
  Systems}, pages 1879--1884, Oct 2012.

\bibitem{dfs}
R.~Tarjan.
\newblock Depth-first search and linear graph algorithms.
\newblock {\em SIAM journal on computing}, 1(2):146--160, 1972.

\bibitem{Tron2}
R.~{Tron} and R.~{Vidal}.
\newblock Distributed image-based 3-d localization of camera sensor networks.
\newblock In {\em {IEEE} International Conference on Decision and Control},
  pages 901--908, Dec 2009.

\bibitem{Tron1}
R.~{Tron} and R.~{Vidal}.
\newblock Distributed 3-d localization of camera sensor networks from 2-d image
  measurements.
\newblock {\em {IEEE} Transactions on Automatic Control}, 59(12):3325--3340,
  Dec 2014.

\bibitem{venables2013modern}
W.~N. Venables and B.~D. Ripley.
\newblock {\em Modern applied statistics with S-PLUS}.
\newblock Springer Science \& Business Media, 2013.

\bibitem{inproceedings}
K.~Wilson and N.~Snavely.
\newblock Robust global translations with 1{DSfM}.
\newblock In {\em {IEEE} European Conference on Computer Vision}, volume 8691,
  pages 61--75, 09 2014.

\bibitem{Yang2013}
Z.~{Yang}, C.~{Wu}, T.~{Chen}, Y.~{Zhao}, W.~{Gong}, and Y.~{Liu}.
\newblock Detecting outlier measurements based on graph rigidity for wireless
  sensor network localization.
\newblock {\em IEEE Transactions on Vehicular Technology}, 62(1):374--383, Jan
  2013.

\bibitem{Zach}
C.~Zach.
\newblock Robust bundle adjustment revisited.
\newblock In {\em {IEEE} European Conference on Computer Vision}, pages
  772--787, 09 2014.

\bibitem{zach2010}
C.~Zach, M.~Klopschitz, and M.~Pollefeys.
\newblock Disambiguating visual relations using loop constraints.
\newblock In {\em {IEEE} Conference on Computer Vision and Pattern
  Recognition}, pages 1426--1433. IEEE, 2010.

\bibitem{LUD}
O.~{Özyeşil} and A.~{Singer}.
\newblock Robust camera location estimation by convex programming.
\newblock In {\em {IEEE} Conference on Computer Vision and Pattern
  Recognition}, pages 2674--2683, June 2015.

\end{thebibliography}

\end{document}